\documentclass[10pt,usletter]{article}

\usepackage[totalwidth=450pt,totalheight=640pt]{geometry}

\usepackage{times}
\usepackage{helvet}
\usepackage{courier}
\usepackage{calc}

\usepackage{amsthm}
\usepackage{amsfonts}
\usepackage{amsmath}
\usepackage{url}
\usepackage{tikz}
\usetikzlibrary{arrows,positioning,automata,decorations,fit,backgrounds}
\usepgflibrary{shapes.geometric} 
\usetikzlibrary{shapes.geometric} 
\usepgflibrary{decorations.pathmorphing} 
\usetikzlibrary{decorations.pathmorphing} 
\usepgflibrary{decorations.text} 
\usetikzlibrary{decorations.text} 
\usetikzlibrary{decorations.shapes} 

\def\hy{\hbox{-}\nobreak\hskip0pt}

\usepackage{natbib}

\newcommand{\citex}[1]{\citeauthor{#1}~(\citeyear{#1})}
\newcommand{\citey}[1]{\citeauthor{#1},~\citeyear{#1}}

\newtheorem{theorem}{Theorem}

\newtheorem{lemma}[theorem]{Lemma}

{}

\newcommand{\W}[1][xxxx]{\text{\normalfont W[#1]}}
\newcommand{\FPT}{\text{\normalfont FPT}}

\newcommand{\iset}{{\cal I}}
\newcommand{\Card}[1]{|#1|}
\newcommand{\SB}{\{\,}%
\newcommand{\SM}{\;{:}\;}%
\newcommand{\SE}{\,\}}%

\newcommand{\true}{\textup{true}}
\newcommand{\false}{\textup{false}}

\begin{document}

\title{On Finding Optimal Polytrees}

\author{
Serge Gaspers \\
The University of New South Wales and\\
Vienna University of Technology \\
{gaspers@kr.tuwien.ac.at} \\
\and
Mikko Koivisto \\
University of Helsinki \\
{mikko.koivisto@cs.helsinki.fi} \\
\and
Mathieu Liedloff \\
Universit\'{e} d'Orl\'{e}ans \\
{mathieu.liedloff@univ-orleans.fr} \\
\and
Sebastian Ordyniak \\
Vienna University of Technology \\
{ordyniak@kr.tuwien.ac.at} \\
\and
Stefan Szeider \\
Vienna University of Technology \\
{stefan@szeider.net} \\
}

\date{}
%


\maketitle

\begin{abstract}
  Inferring probabilistic networks from data is a notoriously difficult
  task. Under various goodness-of-fit measures, finding an optimal
  network is NP-hard, even if restricted to polytrees of
  bounded in-degree. Polynomial-time algorithms are known only for rare
  special cases, perhaps most notably for branchings, that is, polytrees
  in which the in-degree of every node is at most one. Here, we study
  the complexity of finding an optimal polytree that can be turned into
  a branching by deleting some number of arcs or nodes, treated as a parameter. 
We show that the problem can be solved via a matroid intersection
formulation in polynomial time if the number of deleted arcs is bounded
by a constant. The order of the polynomial time bound depends on this
constant, hence the algorithm does not establish fixed-parameter
tractability when parameterized by the number of deleted arcs. We show
that a restricted version of the problem allows fixed-parameter
tractability and hence scales well with the parameter. We contrast this
positive result by showing that if we parameterize by the number of
deleted nodes, a somewhat more powerful parameter, the problem is not
fixed-parameter tractable, subject to a complexity-theoretic assumption.
\end{abstract}

\section{Introduction}

There has been extensive research on learning probabilistic networks
from data by maximizing some suitable scoring function. 
\citex{Edmonds67} gave an efficient algorithm for the class of {\em
  branchings}, that is, directed forests with in-degree at most one; the
algorithm was discovered independently by \citex{ChuL65}, and
it has been later simplified and expedited by others
\cite{Bock71,CameriniFM79,Fulkerson74,GabowGST86,GabowGS89,Karp71,Tarjan77}.
\citex{Chickering95} showed that for general directed acyclic graphs,
DAGs, the problem is NP-hard even if the in-degree is at most
two. Motivated by this gap, \citex{Dasgupta99} asked for a
network class that is more general than branchings yet admitting
provably good structure-learning algorithms; his findings concerning
{\em polytrees}, that is, DAGs without undirected cycles, were however
rather negative, showing that the optimization problem is NP-hard even
if the in-degree is at most two.

Given the recent advances in exact exponential algorithms in general
(see, e.g., the book by \citex{FominK10}), and in finding
optimal DAGs in particular, it is natural to ask, whether ``fast''
exponential-time algorithms exist for finding optimal polytrees. For
general DAGs the fastest known algorithms run in time within a
polynomial factor of $2^n$, where $n$ is the number of nodes
\cite{KoivistoSood04,OttM03,ParviainenK09,SilanderMyllymaki06}.
However, it is not clear, whether even these bounds can be achieved for
polytrees; a brute-force algorithm would visit each polytree one by one,
whose number scales as the number of directed labelled trees
$n^{n-2}2^{n-1}$ \cite{Caley1889}. Do significantly faster algorithms
exist? Does the problem become easier if only a small number of nodes
are allowed an in-degree larger than one?

In this work, we take a first step towards answering these questions by
considering polytrees that differ from branchings by only a few
arcs. More precisely, we study the problem of finding an optimal
{\em $k$-branching}, defined as a polytree that can be turned into a branching by deleting $k$ arcs. We make the standard assumption that the scoring function decomposes into a sum of local scores; see the next section for precise definitions. We note that $k$-branchings generalize branchings in a different direction than the Tree-augmented Naive Bayes classifier (TAN) due to \citex{Friedman1997}. Namely, in a TAN the in-degree of each node is at most two, and there is a designated class node of in-degree zero, removing of which leaves a spanning tree; the tree is undirected in the sense that the symmetric conditional mutual information is employed to score arcs.

\paragraph{Polynomial-time result for $k$-branchings}

Our main result is an algorithm that finds an optimal $k$\hy branching
in polynomial time for every constant $k$. (See the next section for a formal
definition of the problem.) Our overall approach is straightforward: 
we search exhaustively over all possible
sets of at most $k$ ``extra arcs'', fix the guessed arcs, and solve the
induced optimization problem for branchings. Implementing this seemingly
innocent algorithm, however, requires successful treatment of certain
complications that arise when applying the existing matroid machinery for
finding optimal branchings. In particular, one needs to control the
interaction of the extra arcs with the solution from the induced
subproblem.

\paragraph{Fixed-parameter tractability}
Our algorithm for the $k$-branching is polynomial for fixed~$k$, but the
degree of the polynomial depends on $k$, hence the algorithm does not
scale well in $k$. We therefore investigate variants of the
$k$-branching problem that admit \emph{fixed-parameter tractability}
in the sense of \citex{DowneyFellows99}: the running time bound is given by a polynomial whose degree is independent of the parameter, the parameter contributing a constant factor to the bound.

In particular, we show that the $k$-branching problem is fixed-parameter
tractable if the set of arcs incident to nodes with more than one parent form a connected polytree with
exactly one sink, and each node has a bounded number of potential parent
sets. This result is interesting as we show that the $k$-branching problem remains
NP-hard under these restrictions.

We complement the fixed-parameter tractability result by showing that
more general variants of the $k$-branching problem are not
fixed-parameter tractable, subject to complexity theoretic
assumptions. In particular, we show that the $k$-branching problem is
not fixed-parameter tractable when parameterized by the \emph{number of
  nodes} whose deletion produces a branching.

\section{The $k$-branching problem}

A probabilistic network is a multivariate probability distribution that obeys a
structural representation in terms of a directed graph and a corresponding collection of univariate conditional probability distributions. For our purposes, it is crucial to treat the directed graph explicitly, whereas the conditional probabilities will enter our formalism only
implicitly. Such a graph is formalized as a pair $(N, A)$, where $N$ is the {\em node set} and $A \subseteq N\times N$ is the {\em arc set}; we identify the graph with the arc set $A$ when there is no ambiguity about the node set. A node $u$ is said to be a {\em parent} of $v$ in the graph if the arc $(u, v)$ is in $A$; we denote by $A_v$ the set of parents of $v$. When our interest is in the undirected structure of the graph, we may denote by $\overline{A}$ the {\em skeleton} of $A$, that is, the set of {\em edges} $\SB\{u,v\} \SM (u,v) \in A\SE$. For instance, we call $A$ a {\em polytree} if $\overline{A}$ is acyclic, and a {\em branching} if additionally each node has at most one parent. 

When learning a probabilistic network from data it is customary to introduce a scoring function that assigns each graph $A$ a real-valued score $f(A)$ that measures how well $A$ fits the data. While there are plenty of alternative scoring functions, derived under different statistical paradigms and assumptions \cite{LamBacchus94,Chickering95b,HeckermanGeigerChickering95,Dasgupta99}, the most popular ones share one important property: they are {\em decomposable}, that is, 
\[
	f(A) = \sum_{v \in N} f_v(A_v)\,, 
\]
with some ``local'' scoring functions $f_v$. The generic computational problem is to maximize the scoring function over some appropriate class of graphs given the local scoring functions as input. Note that the score $f_v(A_v)$ need not be a sum of any individual arc weights, and that the parent set $A_v$ may be empty.
Figure~\ref{fig:score+branching} shows a table representing a local
scoring function $f$, together with an optimal polytree.

\begin{figure}
  \centering
  \begin{tikzpicture}
  
    \tikzstyle{every circle node}=[inner sep=10pt,draw]
    \tikzstyle{every edge}=[draw,line width=8pt]
    
    \begin{scope}[scale=1]
      \begin{scope}
        \draw (0,0) node[rectangle, inner sep=0pt] {
          $\begin{array}{cc|c}
            v & P & f_v(P) \\
            \hline
            3 & \{1\} & 1.0 \\
            4 & \emptyset & 0.1 \\
            4 & \{1\} & 0.2 \\
            5 & \{1\} & 0.5 \\
            5 & \{1,2\} & 1.0 \\
            6 & \{3\} & 0.8 \\
            6 & \{3,4\} & 1.0 \\
            7 & \{5\} & 0.9 \\
            7 & \{4,5\} & 1.0
          \end{array}$
        };
        
      \end{scope}
      
      \draw
      (2cm,0cm) node {\Large $\mapsto$}
      ;
      
      \tikzstyle{every circle node}=[inner sep=2pt,draw]
      \tikzstyle{every edge}=[draw,line width=1pt]
      
      \begin{scope}[xshift=4.7cm,yshift=1.5cm,scale=.4]
        \draw
        (-2cm,0) node[circle,label=above:$1$] (n1) {}
        (2cm,0) node[circle,label=above:$2$] (n2) {}
        
        (-4cm,-4cm) node[circle,label=left:$3$] (n3) {}
        (0cm,-4cm) node[circle,label=right:$4$] (n4) {}
        (4cm,-4cm) node[circle,label=right:$5$] (n5) {}
        
        (-2cm,-8cm) node[circle,label=below:$6$] (n6) {}
        (2cm,-8cm) node[circle,label=below:$7$] (n7) {}
        ;
        
        \draw
        (n1) edge[-latex] (n3)
        (n1) edge[-latex] (n5)
        
        (n2) edge[-latex] (n5)
        
        (n3) edge[-latex] (n6)
        (n4) edge[-latex] (n6)
        
        (n5) edge[-latex] (n7)
        ;
      \end{scope}
                                            

            \end{scope}
          \end{tikzpicture}
          \caption{An optimal polytree for a given scoring
            function.}
  \label{fig:score+branching}
\end{figure}
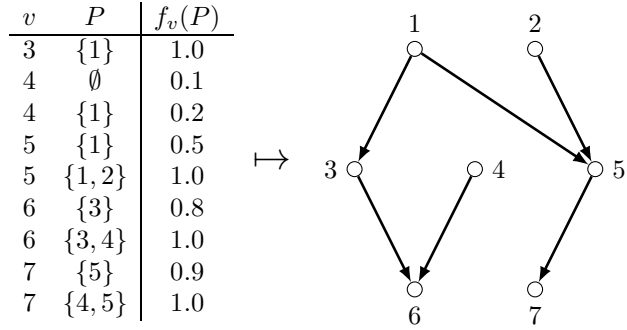

We study this problem by restricting ourselves to a graph class that is
a subclass of polytrees but a superclass of branchings. We call a
polytree $A$ a {\em $k$-branching} if there exists a set of at most $k$
arcs $D \subseteq A$ such that in $A\setminus D$ every node has at most
one parent. Note that any branching is a $0$-branching. The {\em
  $k$-branching problem} is to find a $k$-branching $A$ that maximizes
$f(A)$, given the values $f_v(A_v)$ for each node $v$ and some
collection of possible parent sets $A_v \subseteq N\setminus\{v\}$.


\section{An algorithm for the $k$-branching problem}

Throughout this section we consider a fixed instance of the
$k$-branching problem, that is, a node set $N$ and scoring functions
$f_v$ for each $v \in N$. Thus all arcs will refer to elements of
$N\times N$. We will use the following additional notation. If $A$ is an
arc set, then $H(A)$ denotes the {\em heads} of the arcs in $A$, that
is, the set $\SB v \SM (u, v) \in A\SE$. If $C$ is a set of edges, then
$N(C)$ denotes the induced node set $\SB u, v \SM \{u, v\} \in C\SE$.

We present an algorithm that finds an optimal $k$-branching by
implementing the following approach. First, we guess an arc set $D$ of
size at most $k$. Then we search for an optimal polytree $A$ that
contains $D$ such that in $A\setminus D$ every node has at most one
parent; in other words, $B = A\setminus D$ is an optimal branching with
respect to an induced scoring function. Clearly, the set $\overline{D}$
must be acyclic. The challenge is in devising an algorithm that finds an
optimal branching $B$ that is disjoint from $D$ while guaranteeing that
the arcs in $D$ will not create undirected cycles in the union $B\cup
D$. To this end, we will employ an appropriate weighted matroid
intersection formulation that extends the standard formulation for
branchings.

We will need some basic facts about matroids.  A {\em matroid} is a pair
$(E, \iset)$, where $E$ is a set of {\em elements}, called the {\em
  ground set}, and $\iset$ is a collection of subsets of $E$, called the
{\em independent sets}, such that

\begin{quote}
(M1) $\emptyset \in \iset$;

(M2) if $A\subseteq B$ and $B\in \iset$ then $A\in \iset$; and

(M3) if $A, B\in \iset$ and $|A| < |B|$ then there exists an $e \in B\setminus A$ such that
$A\cup\{e\} \in \iset$.
  
\end{quote}

The {\em rank} of a matroid is the cardinality of its maximal independent sets.
Any subset of $E$ that is not independent is called {\em dependent}.
Any minimal dependent set is called a {\em circuit}.

The power of matroid formulations is much due to the availability of efficient algorithms \cite{BrezovecCG86,Edmonds70,Edmonds79,Frank81,IriT76,Lawler76} for  
the {\em weighted matroid intersection problem}, defined as follows. Given two matroids
$M_1 = (E, \iset_1)$ and $M_2 = (E, \iset_2)$, and a weight function $w \colon E \rightarrow \mathbb{R}$, find an $I\subseteq E$ that is independent in both matroids and maximizes the total weight of $I$, that is, $w(I)= \sum_{e \in I} w(e)$. The complexity of the fastest algorithm we are aware of (for the general problem) is summarized as follows.

\begin{theorem}[\citey{BrezovecCG86}]\label{the:brezovec}
The weighted matroid intersection problem can be solved in $O(m r (r + c + \log m))$ time, where 
$m = |E|$, $r$ is the minimum of the ranks of $M_1$ and $M_2$, and $c$ is the time needed for finding the circuit of $I \cup \{e\}$ in both $M_1$ and $M_2$ where $e \in E$ and $I$ is independent in both $M_1$ and $M_2$. 
\end{theorem}

We now proceed to the specification of two matroids, $M_1(S) = (N\times N, \iset_1(S))$ and 
$M_2(N\times N, \iset_2(S))$, parametrized by an arbitrary arc set $S$ such that $\overline{S}$ is acyclic.
\begin{description}
\item{The {\em in-degree matroid} $M_1(S)$:}
Let $\iset_1(S)$ consist of all arc sets $B$ such that no arc in $B$ has a head in $H(S)$ and every node outside $H(S)$ is the head of at most one arc in $B$.
\item{The {\em acyclicity matroid} $M_2(S)$:}
Let $\iset_2(S)$ consist of all arc sets $B$ such that $\overline{B\cup S}$ is acyclic. 
\end{description}
We observe that the standard matroid intersection formulation of branchings is obtained as the special case of $S=\emptyset$: then an arc set is seen to be branching if and only if it is independent in both the in-degree matroid and the acyclicity matroid.  

The next two lemmas show that $M_1(S)$ and $M_2(S)$ are indeed matroids whenever $\overline{S}$ is acyclic. 

\begin{lemma}\label{lem:matroid-one} $M_1(S)$ is a matroid.
\end{lemma}
\begin{proof}
Fix the arc set $S$ and denote $\iset_1(S)$ by $\iset_1$ for short. 
  Clearly, $\emptyset \in \iset_1$ and if $A \subseteq B$ and $B \in \iset_1$ then
  also $A \in \iset_1$. Consequently, $M_1(S)$ satisfies (M1) and (M2). To see that
  $M_1(S)$ satisfies (M3) let $A,B \in \iset_1$ with $|A| < |B|$. Because of the
  definition of $M_1(S)$ the sets $A$ and $B$ contain at most one arc with head $v$, for every
  $v \in N \setminus H(S)$. Because $|A| < |B|$ there is a node
  $v \in N \setminus H(S)$ such that $v$ is the head of an arc in
  $B$ but $v$ is not the head of an arc in $A$. Let $e \in B$ be the
  arc with head $v$. Then $e \in B \setminus A$ and $A \cup \{e\} \in
  \iset_1$. Hence, $M_1(S)$ satisfies (M3).
\end{proof}

\begin{lemma}\label{lem:matroid-two} $M_2(S)$ is a matroid. 
\end{lemma}
\begin{proof}
  Fix the arc set $S$ and denote $\iset_2(S)$ by $\iset_2$ for short. 
  Because the skeleton 
  $\overline{S}$ is acyclic and acylicity is a hereditary property
  (a graph property is called hereditary if it is closed under taking induced subgraphs)
  it follows that $\emptyset \in \iset_2$ and if $A \subseteq B$ and $B \in \iset_2$ then
  also $A \in \iset_2$. Consequently, $M_2(S)$ satisfies (M1) and (M2). To see that
  $M_2(S)$ satisfies (M3) let $A,B \in \iset_2$ with $|A| < |B|$. 
Consider the sets  $A' = \overline{A \cup S}$ and $B' = \overline{B \cup S}$. Let $C$ be a connected subset of $A'$. Because both $A'$ and $B'$ are acyclic, 
  it follows that the number of edges of $B'$ with both endpoints in
  $N(C)$ is at most the number of edges of $A'$ with both endpoints
  in $N(C)$. Because every edge in $A' \setminus \overline{S}$ corresponds to
  an arc in $A$ and similarly every edge in $B' \setminus \overline{S}$
  corresponds to an arc in $B$ and $|A|<|B|$, it follows that there is
  an arc $e \in B \setminus A$ whose endpoints are contained in two
  distinct components of $A'$. Consequently, the set $A' \cup
  \overline{\{e\}}$ is acyclic and hence $A \cup \{e\} \in \iset_2$.
\end{proof}

We now relate the common independent sets of these two matroids to $k$-branchings. If $A$ is a $k$-branching, we call an arc set $D$ a {\em deletion set} of $A$ if $D$ is a subset of $A$, contains at most $k$ arcs, and in $A\setminus D$ every node has at most one parent.

\begin{lemma}\label{lem:matroid-three}
  Let $A$ be an arc set and  $D$ a subset of $A$
  of size at most $k$ such that no two arcs from
  $D'=\SB (u,v)\in A\setminus D \SM v\in H(D)\SE$
  have the same head and such that $\overline{S}$ is acyclic, where $S=D\cup D'$.
  We have that $A$ is a $k$-branching with deletion set $D$ if and only if
  $A\setminus S$ is independent in both $M_1(S)$ and $M_2(S)$.
\end{lemma}
\begin{proof}
  $(\Rightarrow):$ Suppose $A$ is a $k$-branching with deletion set $D$.
  Then $A\setminus D$ is a branching, which shows that every node $v$
  outside $H(S)$ has in-degree at most one in $A\setminus S$. Since by
  definition all arcs with a head in $H(S)$ are contained in $S$, no arc
  in $A\setminus S$ has a head in $H(S)$. Therefore, $A\setminus S$ is
  independent in $M_1(S)$.  Since every $k$-branching is a polytree,
  $\overline{(A\setminus S) \cup S} = \overline{A}$ is acyclic, and
  therefore $A\setminus S$ is independent in $M_2(S)$.

 $(\Leftarrow):$ Since $A\setminus S$ is independent in $M_2(S)$, we
 have that $\overline{(A\setminus S) \cup S} = \overline{A}$ is
 acyclic. Thus, $A$ is a polytree. As $A\setminus S$ is independent in
 $M_1(S)$, every node outside $H(S)$ has in-degree at most one in
 $A\setminus S$ and every node from $H(S)$ has in-degree zero in
 $A\setminus S$. Since the head of every arc from $D'$ is in $H(S)$ and
 no two arcs from $D'$ have a common head, $(A\setminus S) \cup D' =
 A\setminus D$ has maximum in-degree at most one. Because $|D|\le k$, we
 have that $A$ is a $k$-branching with deletion set $D$.
\end{proof}

The characterization of Lemma~\ref{lem:matroid-three} enables the following algorithm for the $k$-branching problem. Define the weight function by letting $w(u, v) = f_v(\{u\}) - f_v(\emptyset)$ for all arcs $(u, v)$. Guess the arc sets $D$ and $D'$, put $S = D\cup D'$, check that $\overline{S}$ is acyclic, find a maximum-weight set $B$ that is independent in both $M_1(S)$ and $M_2(S)$; output a $k$-branching $A = B\cup S$ that yields the maximum weight over all guesses $D$ and $D'$, where the weight of $B\cup S$ is obtained as
\[
	w(B) + \sum_{v \in H(S)} \big(f_v(S_v) - f_v(\emptyset)\big)\,.
\]
It is easy to verify that maximizing this weight is equivalent to maximizing the score $f(A)$. 
Figure~\ref{fig:algo} illustrates the algorithm for the scoring function
of Figure~\ref{fig:score+branching}.

\begin{figure}
  \centering
  \definecolor{Blue}{rgb}{0.0,0.6,1}
  \begin{tikzpicture}[scale=.35]
    \tikzstyle{every circle node}=[inner sep=2pt,draw]
    \tikzstyle{every edge}=[draw,line width=1pt]
    \draw
    (-2cm,0) node[circle,label=above:$1$] (n1) {}
    (2cm,0) node[circle,label=above:$2$] (n2) {}
    
    (-4cm,-4cm) node[circle,label=left:$3$] (n3) {}
    (0cm,-4cm) node[circle,label=right:$4$] (n4) {}
    (4cm,-4cm) node[circle,label=right:$5$] (n5) {}
    
    (-2cm,-8cm) node[circle,label=below:$6$] (n6) {}
    (2cm,-8cm) node[circle,label=below:$7$] (n7) {}
    ;
    
    \draw
    (n3) edge[-latex,densely dotted] (n6)
    (n4) edge[-latex,dashed] (n6)
    
    (n1) edge[-latex,densely dotted] (n5)
    (n2) edge[-latex,dashed] (n5)
    ;
  \end{tikzpicture}\hspace{4cm}\begin{tikzpicture}[scale=.35]

    \tikzstyle{every circle node}=[inner sep=2pt,draw]
    \tikzstyle{every edge}=[draw,line width=1pt]

    \draw
    (-2cm,0) node[circle,label=above:$1$] (n1) {}
    (2cm,0) node[circle,label=above:$2$] (n2) {}
    
    (-4cm,-4cm) node[circle,label=left:$3$] (n3) {}
    (0cm,-4cm) node[circle,label=right:$4$] (n4) {}
    (4cm,-4cm) node[circle,label=right:$5$] (n5) {}
    
    (-2cm,-8cm) node[circle,label=below:$6$] (n6) {}
    (2cm,-8cm) node[circle,label=below:$7$] (n7) {}
    ;
    
    \draw
    (n3) edge[-latex,densely dotted] (n6)
    (n4) edge[-latex,dashed] (n6)
    
    (n1) edge[-latex,densely dotted] (n5)
    (n2) edge[-latex,dashed] (n5)
    
    (n1) edge[-latex] (n3)
    (n5) edge[-latex] (n7)
    ;
  \end{tikzpicture}

\newcommand{\cD}[1]{{\color{red} #1}}
\newcommand{\cDP}[1]{{\color{violet} #1}}
\newcommand{\cB}[1]{{\color{Blue} #1}}

  \caption{Left: the two guessed arc sets $D$ (dotted) and $D'$ (dashed).
    Right: the arc set $A$ (solid) that is a heaviest common independent
    set of the two matroids $M_1(S)$ and $M_2(S)$.}
\label{fig:algo}
\end{figure}
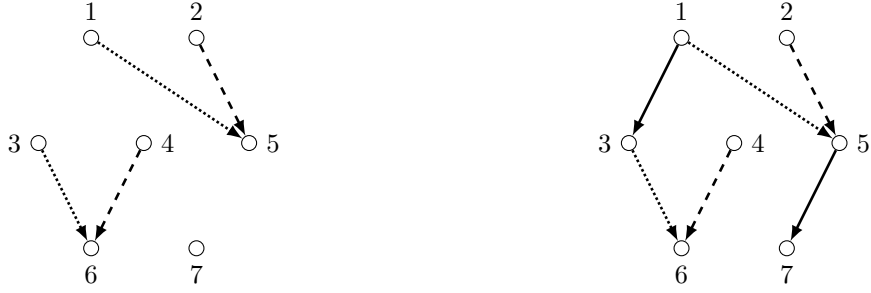

It remains to analyze the complexity of the algorithm. Denote by $n$ the
number of nodes.  For a moment, consider the arc set $S$ fixed. To apply
Theorem~\ref{the:brezovec}, we bound the associated key quantities: the
size of the ground set is $O(n^2)$; the rank of both matroids is clearly
$O(n)$; circuit detection can be performed in $O(n)$ time, by a
depth-first search for $M_1(S)$ and by finding a node that has higher
in-degree than it is allowed to have in $M_2(S)$. Thus, by
Theorem~\ref{the:brezovec}, a maximum-weight set that is independent in
both matroids can be found in $O(n^4)$ time. Then consider the number of
possible choices for the set $S = D\cup D'$. There are $O(n^{2k})$
possibilities for choosing a set $D$ of at most $k$ arcs such that
$\overline{D}$ is acyclic. For a fixed $D$, there are $O(n^{k})$
possibilities for choosing a subset $D' \subseteq N\times H(D)$ such
that $\overline{D\cup D'}$ is acyclic and no two arcs from $D'$ have the
same head. Thus there are $O(n^{3k})$ relevant choices for the set $S$.

We have shown the following. 
\begin{theorem}\label{the:xp} 
The $k$-branching problem can be solved in $O(n^{3k+4})$ time.
\end{theorem}

\section{Fixed-parameter tractability}

\label{sec:fpt}

Theorem~\ref{the:xp} shows that the $k$\hy branching problem can be
solved in ``non-uniform polynomial time'' as the order of the polynomial
time bound depends on~$k$. In this section we study the question of whether
one can get $k$ ``out of the exponent'' and obtain a uniform
polynomial-time algorithm.

The framework of \emph{Parameterized Complexity} \cite{DowneyFellows99}
offers the suitable tools and methods for such an investigation, as it
allows us to distinguish between uniform and non-uniform polynomial-time
tractability with respect to a parameter. An instance of a parameterized
problem is a pair $(I,k)$ where $I$ is the \emph{main part} and $k$ is
the \emph{parameter}; the latter is usually a non-negative integer.  A
parameterized problem is \emph{fixed-parameter tractable} if there exist
a computable function $f$ and a constant~$c$ such that instances $(I,k)$
of size $n$ can be solved in time $O(f(k)n^c)$. $\FPT$ is the class of
all fixed-parameter tractable decision problems. Fixed-parameter
tractable problems are also called \emph{uniform polynomial-time
  tractable} because if $k$ is considered constant, then instances with
parameter $k$ can be solved in polynomial time where the order of the
polynomial is independent of $k$ (in contrast to non-uniform
polynomial-time running times such as $n^k$).

Parameterized complexity offers a completeness theory similar to the
theory of NP-completeness.  One uses \emph{parameterized reductions}
which are many-one reductions where the parameter for one problem maps
into the parameter for the other. More specifically, problem $L$ reduces
to problem $L'$ if there is a mapping $R$ from instances of $L$ to
instances of $L'$ such that (i)~$(I,k)$ is a yes-instance of $L$ if and
only if $(I',k')=R(I, k)$ is a yes-instance of~$L'$, (ii)~$k'\leq g(k)$ for
a computable function $g$, and (iii)~$R$ can be computed in time
$O(f(k)n^c)$ where $f$ is a computable function, $c$ is a constant, and
$n$ denotes the size of $(I,k)$.  The parameterized complexity class
$\W[1]$ is considered as the parameterized analog to NP. For example,
the parameterized Maximum Clique problem (given a graph $G$ and a
parameter $k\geq 0$, does $G$ contain a complete subgraph on $k$
vertices?) is $\W[1]$-complete under parameterized reductions. Note that
there exists a trivial non-uniform polynomial-time $n^k$ algorithm for
the Maximum Clique problems that checks all sets of $k$ vertices.
$\FPT\neq \W[1]$ is a widely accepted complexity theoretic assumption
\cite{DowneyFellows99}.  For example, $\FPT=\W[1]$ implies the (unlikely)
existence of a $2^{o(n)}$ algorithm for $n$\hy variable
3SAT~\cite{ImpagliazzoPaturiZane01,FlumGrohe06}. A first parameterized
analysis of probabilistic network structure learning using structural
parameters such as treewidth has recently been carried out
by~\citex{OrdyniakSzeider10}.

The algorithm from Theorem \ref{the:xp} considers
$O(n^{3k})$ relevant choices for the set $S = D\cup D'$,
and for each fixed choice of $S$ the running time is polynomial.
Thus, for restrictions of the problem for which the enumeration
of all relevant sets $S$ is fixed parameter tractable, one obtains
an FPT algorithm. One such restriction requires that $S=D \cup D'$ is an in-tree,
i.e., a directed tree where every arc is directed towards a designated root,
and each node has a bounded number of potential parent sets.

\begin{theorem}\label{thm:fptrestriction}
  The $k$-branching problem is fixed-parameter tractable if we require
  that (i)~the set $S=D\cup D'$ of arcs is an in-tree
  and  (ii)~each node has a bounded number of potential parent sets. 
\end{theorem}
\begin{proof}
 To compute a $k$-branching $A$, the algorithm guesses its deletion set
 $D$ and the set $D' = \{ (u,v)\in A \setminus D: v\in H(D)\}$.
 As $A$ is a $k$-branching, $|D|\le k$ and for every $v\in H(D)$
 there is at most one arc in $D'$ with head $v$.
 The algorithm first guesses the root $r$ for the in-tree $S$.
 Then it goes over all possible choices for $D$ and $D'$ as follows,
 until $D$ has at least $k$ arcs.
 
 Guess a leaf $\ell$ of $S$ (initially, $r$ is the unique leaf of $S$), and guess
 a non-empty parent set $P$ for $\ell$ in $A$. If $|D|+|P|+1>k$, then backtrack.
 Otherwise, choose at most one arc $(p,\ell)$ to add to $D'$, where $p\in P$, and add all other arcs
 from a node from $P$ to $\ell$ to $D$ (if $|P|=1$, no arc is added to $D'$).
 Now, check whether the current choice for $S=D\cup D'$ leads to a $k$-branching
 by checking whether $\overline{S}$ is acyclic and
 using the matroids $M_1(S)$ and $M_2(S)$ as in Theorem \ref{the:xp}.

 There are at most $n$ choices for $r$. The in-tree $S$ is expanded in at most
 $k$ steps, as each step adds at least one arc to $D$. In each step, $\ell$ is chosen
 among at most $k+1$ leaves, there is a constant number of choices for its parent set $P$ and
 at most $k+2$ choices for adding (or not) an arc $(p,\ell)$, with $p\in P$, to $D'$ (as $|P|\le k+1$).
 The acyclicity check for $\overline{S}$ and the weighted matroid intersection
 can be computed in time $O(n^4)$, leading to a total running time of
 $O(k^{2k}c^kn^5)$, where $c$ is such that every node has at most $c$ potential parent sets.
%
\end{proof}
Condition (i) in Theorem \ref{thm:fptrestriction} may be replaced
by other conditions requiring the connectivity of $D$ or a small distance
between arcs from $D$, giving other fixed-parameter tractable
restrictions of the $k$-branching problem.

The following theorem shows that an exponential dependency on $k$
or some other parameter is necessary since 
the $k$-branching problem remains NP-hard under the restrictions
given above.
\begin{theorem}\label{thm:fptrestrictionnp}
  The $k$-branching problem is NP-hard even if we require that
   (i)~the set $S=D\cup D'$ of arcs is an in-tree
  and  (ii)~each node has at most $3$ potential parent sets. 
\end{theorem}
\begin{proof}
  We devise a polynomial reduction from
  \textsc{$3$-SAT-$2$} a version of
  \textsc{3-SATISFIABILITY} where every literal occurs at most in two
  clauses. \textsc{$3$-SAT-$2$} is well known to be
  NP-hard~\cite{GareyJohnson79}. Our reduction uses the same ideas as
  the proof of Theorem~6 in~\cite{Dasgupta99}. Let $\Phi$ be an instance of
  \textsc{$3$-SAT-$2$} with clauses $C_1,\dotso,C_m$ and variables
  $x_1,\dotso,x_n$.   We define the set $N$ of nodes as follows. For every variable $x_i$
  in $\Phi$ the set $N$ contains the nodes
  $p_i,x_i,x_i^1,x_i^2,\overline{x}_i^1$ and $\overline{x}_i^2$. Furthermore, for
  every clause $C_j$ the set $N$ contains the nodes $p_{n+j}$ and
  $C_j$. Let $1 \leq i \leq n$, $1 \leq j \leq m$, and $1\leq l
  \leq 2$. We set $f(C_j,x_i^l)=1$ if the clause $C_j$ is the $l$-th
  clause that contains the literal $x_i$. Similarly, we set
  $f(C_j,\overline{x}_i^l)=1$ if the clause $C_j$ is the $l$-th
  clause that contains the literal $\overline{x}_i$. We set
  $f(x_i,\{x_i^1,x_i^2\})=f(x_i,\{\overline{x}_i^1,\overline{x}_i^2\})=1$,
  $f(p_{1},\{x_{1}\})=f(p_{i},\{x_{i},p_{i-1}\})=1$ for every $1 < i
  \leq n$, and
  $f(p_{n+j},\{C_j,p_{n+j-1}\})=1$ for every
  $1 \leq j \leq m$. Furthermore, we set $f(v,P)=0$ for all the remaining
  combinations of $v \in N$ and $P \subseteq N$. This completes the
  construction of $N$ and $f$. Observe 
  that every node of $N$ has at most $3$ potential parent sets.
  This completes our construction. We will have shown the theorem
  after showing the following claim.

  \emph{Claim: $\Phi$ is satisfiable if and only if there is a
    $2n+m$-branching~$D$ such that $f(D)\geq 2(m+n)$, the set $S=D \cup D'$ of arcs is an in-tree, and 
    each node of $N$ has at most $3$ potential parent sets.  }

  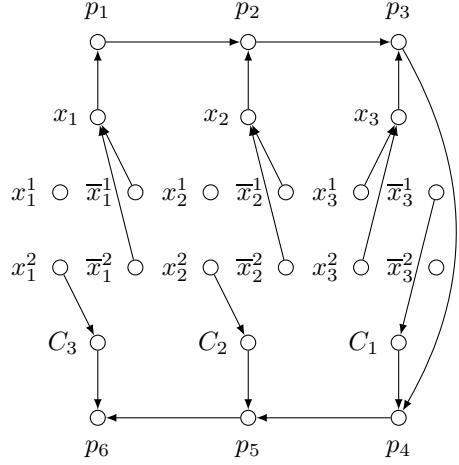
\begin{figure}
    \centering
    \begin{tikzpicture}[scale=1]
      \tikzstyle{every node}=[circle,inner sep=2pt,draw,color=black]
      \tikzstyle{every edge}=[draw]
      
      \begin{scope}
        \draw
        (0,0) node[label=above:$p_1$] (p1) {}
        (2,0) node[label=above:$p_2$] (p2) {}
        (4,0) node[label=above:$p_3$] (p3) {}

        (0,-1) node[label=left:$x_1$] (x1) {}
        (-0.5,-2) node[label=left:$x_1^1$] (x11) {}
        (-0.5,-3) node[label=left:$x_1^2$] (x12) {}
        (0.5,-2) node[label=left:$\overline{x}_1^1$] (nx11) {}
        (0.5,-3) node[label=left:$\overline{x}_1^2$] (nx12) {}

        (2,-1) node[label=left:$x_2$] (x2) {}
        (1.5,-2) node[label=left:$x_2^1$] (x21) {}
        (1.5,-3) node[label=left:$x_2^2$] (x22) {}
        (2.5,-2) node[label=left:$\overline{x}_2^1$] (nx21) {}
        (2.5,-3) node[label=left:$\overline{x}_2^2$] (nx22) {}

        (4,-1) node[label=left:$x_3$] (x3) {}
        (3.5,-2) node[label=left:$x_3^1$] (x31) {}
        (3.5,-3) node[label=left:$x_3^2$] (x32) {}
        (4.5,-2) node[label=left:$\overline{x}_3^1$] (nx31) {}
        (4.5,-3) node[label=left:$\overline{x}_3^2$] (nx32) {}

        (0,-4) node[label=left:$C_3$] (c3) {}
        (2,-4) node[label=left:$C_2$] (c2) {}
        (4,-4) node[label=left:$C_1$] (c1) {}

        (4,-5) node[label=below:$p_4$] (p4) {}
        (2,-5) node[label=below:$p_5$] (p5) {}
        (0,-5) node[label=below:$p_6$] (p6) {}
        ;
        
        \draw
        (p1) edge[-latex] (p2)
        (p2) edge[-latex] (p3)
        (p3) edge[-latex, bend left] (p4)
        (p4) edge[-latex] (p5)
        (p5) edge[-latex] (p6)

        (x1) edge[-latex] (p1)
        (x2) edge[-latex] (p2)
        (x3) edge[-latex] (p3)

        (c1) edge[-latex] (p4)
        (c2) edge[-latex] (p5)
        (c3) edge[-latex] (p6)

        (nx11) edge[-latex] (x1)
        (nx12) edge[-latex] (x1)

        (nx21) edge[-latex] (x2)
        (nx22) edge[-latex] (x2)

        (x31) edge[-latex] (x3)
        (x32) edge[-latex] (x3)

        (x12) edge[-latex] (c3)
        (x22) edge[-latex] (c2)
        (nx31) edge[-latex] (c1)
        ;
      \end{scope}
    \end{tikzpicture}
    \caption{An optimal $2n+m$-branching $D$ for the formula
      $\Phi=C_1 \land C_2 \land C_3$ with $C_1=x_1\lor x_2 \lor
      \overline{x}_3$, $C_2=\overline{x}_1\lor x_2 \lor x_3$, and $C_3=x_1\lor \overline{x}_2
      \lor \overline{x}_3$ according to the construction given in the proof of Theorem~\ref{thm:fptrestrictionnp}.}
    \label{fig:np-hard-branching}
  \end{figure}
  $(\Rightarrow):$ Suppose that the formula $\Phi$
  is satisfiable and let $\beta$ be a satisfying assignment for
  $\Phi$. Furthermore, for every $1 \leq j
  \leq m$ let $l_j$ be a literal of $C_j$ that is set to true by
  $\beta$. We construct a $2n+m$-branching $D$ as follows. For every $1\leq j
  \leq m$ the digraph $D$ contains an arc $(x_i^l,C_j)$ if $l_j=x_i$ and $C_j$ is the $l$-th
  clause that contains $x_i$ and an arc
  $(\overline{x}_i^l,C_j)$ if $l_j=\overline{x}_i$
  and $C_j$ is the $l$-th clause that contains $\overline{x}_i$ for
  some $1 \leq i \leq n$ and $1 \leq l \leq 2$. Furthermore, for every $1 \leq i \leq n$ the
  digraph $D$ contains the arcs $(x_i^1,x_i)$ and $(x_i^2,x_i)$ if
  $\beta(x_i)=\false$ and the arcs $(\overline{x}_i^1,x_i)$ and $(\overline{x}_i^2,x_i)$ if
  $\beta(x_i)=\true$. Last but not least $D$ contains the arcs
  $(x_i,p_i)$, $(C_j,p_{n+j})$ and $(p_l,p_{l+1})$ for every $1 \leq i
  \leq n$, $1 \leq j \leq m$, and $1 \leq l <
  m+n$. Figure~\ref{fig:np-hard-branching} shows an optimal
  $2n+m$-branching $D$ for some $3$-SAT-$2$ formula.
  It is easy to see
  that $D$ is a $2n+m$-branching such that $f(D)=2(m+n)$ 
  and the set $S=D\cup D'$ of arcs is an in-tree.

  $(\Leftarrow):$ Suppose there is a 
  $2n+m$-branching $D$ such that $f(D)\geq 2(m+n)$. 
  Because $f(D) \geq 2(m+n)$ it follows
  that every node of $N$ achieves its maximum score in $D$. Hence, $D$
  has to contain the arcs $(x_i,p_i)$, $(C_j,p_{n+j})$,
  $(p_l,p_{l+1})$, for every $1 \leq i \leq n$, $1\leq j \leq m$, and
  $1\leq l <m+n$. For the same reasons $D$ has to contain either the
  arcs $(x_i^1,x_i)$ and $(x_i^2,x_i)$ or the arcs
  $(\overline{x}_i^1,x_i)$ and $(\overline{x}_i^2,x_i)$ for every $1
  \leq i \leq n$. Furthermore, for every $1 \leq j \leq m$ the
  $2n+m$-branching $D$ has to contain one arc of the form $(x_i^l,C_j)$ or
  $(\overline{x}_i^l,C_j)$ where $C_j$ is the $l$-th
  clause that contains $x_i$ or $\overline{x}_i$, respectively, for some $1 \leq i \leq n$ and $1
  \leq l \leq 2$. Let $1 \leq i \leq n$, $1 \leq j
  \leq m$, and $1 \leq l \leq 2$. We first show
  that whenever $D$ contains an arc $(x_i^l,x_i)$ then $D$ contains no arc of
  the form $(x_i^l,C_j)$ and similarly if $D$ contains an arc
  $(\overline{x}_i^l,x_i)$ then $D$ contains no arc of the
  form $(\overline{x}_i,C_j)$. Suppose for a contradiction that $D$
  contains an arc $(x_i^l,x_i)$ together with an arc $(x_i^l,C_j)$ or
  an arc $(x_i^l,x_i)$ together with an arc $(x_i^l,C_j)$.
  In the first case $D$ contains the undirected cycle
  $(x_i^l,x_i,p_i,\dotso,p_{n+j},C_j,x_i^l)$ and in the second case
  $D$ contains the cycle $(\overline{x}_i^l,x_i,p_i,\dotso,p_{n+j},C_j,\overline{x}_i^l)$
  contradicting our assumption that $D$ is a $2n+m$-branching. 
  It now follows that the assignment $\beta$ with
  $\beta(x_i)=\true$ if $D$ does not contain the arcs $(x_i^1,x_i)$
  and $(x_i^2,x_i)$ and $\beta(x_i)=\false$ if $D$ does not contain
  the arcs $(\overline{x}_i^1,x_i)$ and $(\overline{x}_i^2,x_i)$ is a
  satisfying assignment for $\Phi$.
\end{proof}
So far we have measured the difference of a polytree to branchings in terms of the number of arcs to be deleted. Next we investigate the consequences of measuring the difference by the number of nodes to be deleted. We call a polytree $A$ a {\em $k$-node branching} if there exists a set of at most $k$ nodes $X \subseteq A$ such that $A\setminus X$ is a
branching. The {\em $k$-node branching problem} is to find a
$k$-node branching $A$ that maximizes $f(A)$.
Clearly every $k$-branching is a $k$-node branching, but
the reverse does not hold. In other words, the $k$-node branching problem
generalizes the $k$-branching problem. 

In the following we show that the $k$-node branching problem is hard
for the parameterized complexity class $\W[1]$; this provides strong evidence that the problem
is not fixed-parameter tractable.

\begin{theorem}\label{thm:hardness-knodeb}
  The $k$-node branching problem is $\W[1]$-hard.
\end{theorem}
\begin{proof}
  We devise a parameterized reduction from the following problem, called
  Partitioned Clique, which is well-known to be $\W[1]$-complete~for
  parameter~$k$~\cite{Pietrzak03}.  The Instance is a $k$\hy partite
  graph $G=(V,E)$ with partition $V_1,\dotso,V_k$ such that
  $\Card{V_i}=n$ for every $1\leq i \leq k$.  The question is whether
  there are nodes $v_1,\dotso,v_k$ such that $v_i\in V_i$ for $1\leq i
  \leq k$ and $\{v_i,v_j\}\in E$ for $1\leq i < j \leq k$? (The graph
  $K=(\{v_1,\dots,v_k\},\SB \{v_i,v_j\} \SM 1\leq i < j \leq k\SE)$ is a
  \emph{$k$-clique} of $G$.)

  Let $G=(V,E)$ be an instance of this problem with partition $V_1,\dots,V_k$,
  $\Card{V_1}=\dots=\Card{V_k}=n$, and parameter $k$.
  Let $k'=\binom{k}{2}+k$, $\alpha=1$, and
  $s=k'\alpha$. 
  Let $A=\SB a_{ij}\SM 1 \leq i < j \leq k\SE$ and
  $A_i=\SB a_{lk} \in A \SM l=i$ or $k=i\SE$ for every $1 \leq i \leq
  k$. Then $N$ is defined as $N=A \cup
  \{c_1,\dotsc, c_k\} \cup \SB v^1,\dotsc,v^k \SM v \in V \SE$. 
  Let $V_i^w=\SB v^1, \dotsc, v^k
  \SM v \in V_i \textup{ and }v \neq w \SE$. We define the score function 
  $f$ as follows.
  We set $f(c_i,A_i \cup
  V_i^w)=\alpha$ for every $1 \leq i \leq k$ and $w \in V_i$, and
  $f(a_{ij},\{u^j,w^i\})=\alpha$ for every $1 \leq i < j \leq k$, $u \in
  V_i$, $w \in V_j$, and $\{u,w\} \in E(G)$.  Furthermore, we set
  $f(v,P)=0$ for all the remaining combinations of $v$ and $P$.
  This completes our construction. We will have the
  theorem after showing the following claim.

  \emph{Claim: $G$ has a $k$-clique if and only if there is a
    $k'$-node branching~$D$ such that $f(D) \geq s$.  }

  \begin{figure}
    \centering
    \begin{tabular}{c|cc}
      \begin{tikzpicture}[scale=0.5]
        \tikzstyle{every circle node}=[circle,inner sep=1pt,draw]
        \begin{scope}[label distance=4pt, scale=1]
          \begin{scope}
            \draw
            (0,2cm)
            +(-1cm,0) node[circle] (v11) {}
            +(0,0) node[circle,fill] (v12) {}
            +(1cm,0) node[circle] (v13) {}

            +(0,0) node[rectangle,minimum width=1.3cm,minimum
            height=0.4cm,draw,rotate around={0:(0,0)}] {}
            ;
            
            \draw[rotate around={-120:(0,0)}]
            (0,2cm)
            +(-1cm,0) node[circle] (v21) {}
            +(0,0) node[circle,fill] (v22) {}
            +(1cm,0) node[circle] (v23) {}

            +(0,0) node[rectangle,minimum width=1.3cm,minimum
            height=0.4cm,draw,rotate around={-120:(0,0)}] {}
            ;
            
            \draw[rotate around={120:(0,0)}]
            (0,2cm)
            +(-1cm,0) node[circle] (v31) {}
            +(0,0) node[circle,fill] (v32) {}
            +(1cm,0) node[circle] (v33) {}

            +(0,0) node[rectangle,minimum width=1.3cm,minimum
            height=0.4cm,draw,rotate around={120:(0,0)}] {}
            ;
      
            \draw
            (v11) edge[] (v23)
            (v12) edge[thick] (v22)
            (v13) edge[] (v21)
            
            (v21) edge[] (v31)
            (v22) edge[thick] (v32)
            
            (v33) edge[] (v13)
            (v32) edge[thick] (v12)
            (v31) edge[] (v11)
            ;
          \end{scope}
        \end{scope}
      \end{tikzpicture}
      
      &
      &
      \begin{tikzpicture}[scale=0.5]
        \tikzstyle{every circle node}=[circle,inner sep=1pt,draw]
        \begin{scope}[label distance=4pt, scale=1]
          \begin{scope}
            \draw
            (0,2cm)
            +(-1cm,0) node[circle] (v111) {}
            +(0,0) node[circle,fill] (v121) {}
            +(1cm,0) node[circle] (v131) {}

            +(-1cm,0.5) node[circle] (v112) {}
            +(0,0.5) node[circle,fill] (v122) {}
            +(1cm,0.5) node[circle] (v132) {}

            +(-1cm,1) node[circle] (v113) {}
            +(0,1) node[circle,fill] (v123) {}
            +(1cm,1) node[circle] (v133) {}

            +(0,2.5) node[circle] (c1) {}

            +(0,1.25) node[rectangle,minimum width=1.3cm,minimum height=1.5cm,draw] {}
            ;

            \draw[rotate around={-120:(0,0)}]
            (0,2cm)
            +(-1cm,0) node[circle] (v211) {}
            +(0,0) node[circle,fill] (v221) {}
            +(1cm,0) node[circle] (v231) {}

            +(-1cm,0.5) node[circle] (v212) {}
            +(0,0.5) node[circle,fill] (v222) {}
            +(1cm,0.5) node[circle] (v232) {}

            +(-1cm,1) node[circle] (v213) {}
            +(0,1) node[circle,fill] (v223) {}
            +(1cm,1) node[circle] (v233) {}

            +(0,2.5) node[circle] (c2) {}

            +(0,1.25) node[rectangle,minimum width=1.3cm,minimum
            height=1.5cm,draw,rotate around={-120:(0,0)}] {}
            ;
 
            \draw[rotate around={120:(0,0)}]
            (0,2cm)
            +(-1cm,0) node[circle] (v311) {}
            +(0,0) node[circle,fill] (v321) {}
            +(1cm,0) node[circle] (v331) {}

            +(-1cm,0.5) node[circle] (v312) {}
            +(0,0.5) node[circle,fill] (v322) {}
            +(1cm,0.5) node[circle] (v332) {}

            +(-1cm,1) node[circle] (v313) {}
            +(0,1) node[circle,fill] (v323) {}
            +(1cm,1) node[circle] (v333) {}

            +(0,2.5) node[circle] (c3) {}

            +(0,1.25) node[rectangle,minimum width=1.3cm,minimum
            height=1.5cm,draw,rotate around={120:(0,0)}] {}
            ;
            
            \draw
            (30:1) node[circle] (a12) {}
            (150:1) node[circle] (a13) {}
            (-90:1) node[circle] (a23) {};

            \draw
            (v111) edge[-latex] (c1)
            (v112) edge[-latex] (c1)
            (v113) edge[-latex] (c1)

            (v131) edge[-latex] (c1)
            (v132) edge[-latex] (c1)
            (v133) edge[-latex] (c1)

            (a12) edge[-latex,bend right=60] (c1)
            (a13) edge[-latex,bend left=60] (c1)

            (v122) edge[-latex,bend left=10] (a12)
            (v123) edge[-latex,bend right=10] (a13)
            ;

            \draw
            (v211) edge[-latex] (c2)
            (v212) edge[-latex] (c2)
            (v213) edge[-latex] (c2)

            (v231) edge[-latex] (c2)
            (v232) edge[-latex] (c2)
            (v233) edge[-latex] (c2)

            (a12) edge[-latex,bend left=60] (c2)
            (a23) edge[-latex,bend right=60] (c2)

            (v221) edge[-latex,bend right=10] (a12)
            (v223) edge[-latex,bend left=10] (a23)
            ;

            \draw
            (v311) edge[-latex] (c3)
            (v312) edge[-latex] (c3)
            (v313) edge[-latex] (c3)

            (v331) edge[-latex] (c3)
            (v332) edge[-latex] (c3)
            (v333) edge[-latex] (c3)

            (a13) edge[-latex,bend right=60] (c3)
            (a23) edge[-latex,bend left=60] (c3)

            (v321) edge[-latex,bend left=10] (a13)
            (v322) edge[-latex,bend right=10] (a23)
            ;

          \end{scope}
        \end{scope}
      \end{tikzpicture}
      \\
      $G$

      &

      &

      $D$
    \end{tabular}
    \caption{An example graph $G$ ($k=3$) together with an optimal $k'$-node branching
      $D$ with $f(D)\geq s$ according to the construction given in the proof
      of Theorem~\ref{thm:hardness-knodeb}.}
    \label{fig:hardnessknodeb}
  \end{figure}
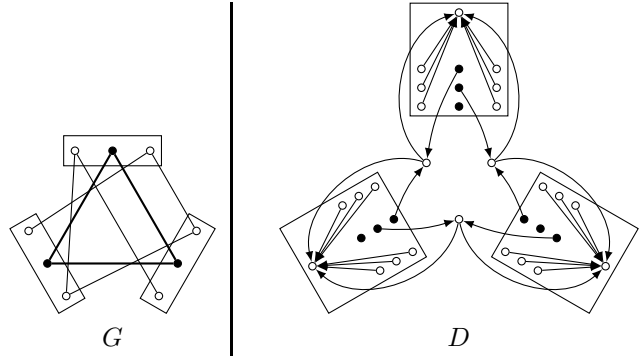
  $(\Rightarrow):$ Suppose that $G$ has a $k$-clique
  $K$. Then it is easy to see that the DAG $D$ on $N$ defined by the arc set 
  $\SB (v^j,a_{ij}),(v^j,a_{ji}) \SM v \in V(K) \cap V_i \textup{ and }1 \leq i, j \leq k \SE \cup
  \SB
  (v^i,a_{ij}),(v^i,a_{ji}) \SM v \in V(K) \cap V_j \textup{ and }1 \leq i,j
  \leq k \SE \cup \SB (a_{ij},c_i) \SM 1 \leq i < j \leq k \SE \cup \SB
  (a_{ij},c_j) \SM 1 \leq i < j \leq k \SE \cup \SB (v^j,c_i) \SM v \in
  V_i \setminus (\bigcup_{e \in E(K)}e) \textup{ and }1 \leq i, j \leq
  k\SE$ is a $k'$-node branching and
  $f(D)=s$. Figure~\ref{fig:hardnessknodeb} shows an optimal $k'$-node
  branching $D$ constructed from an example graph $G$.

  $(\Leftarrow):$ Suppose there is a $k'$-node branching $D$
  with $f(D) \geq s$. It follows that every node of $D$ achieves its
  maximum score. In particular, for every $1 \leq i \leq k$ the nodes
  $c_i$ must have score $\alpha$ in $D$ and hence there is a node $w_i
  \in V_i$ such that $c_i$ is adjacent to all nodes in $V_i^{w_i} \cup
  A_i$. Furthermore, for every $1 \leq i < j \leq k$ the node $a_{ij}$
  is adjacent to exactly one node in $V_i$ and to exactly one node in
  $V_j$. Let $v_i^l$ be the unique node in $V_i$ adjacent to $a_{ij}$
  and similarly let $v_i^m$ be the unique node in $V_j$ that is adjacent
  to $a_{ij}$ for every $1 \leq i < j \leq k$. Then $w_i=v_i$ and
  $w_j=v_j$ because otherwise the skeleton of $D$ would contain the
  cycle $(v_i,a_{ij},c_i)$ or the cycle
  $(v_j,a_{ij},c_j)$. Consequently, the edges represented by the parents
  of $a_{ij}$ in $D$ for all $1 \leq i < j \leq k$ form a $k$-clique in
  $G$.
\end{proof}

\section{Concluding remarks}

We have studied a natural approach to extend the known efficient
algorithms for branchings to polytrees that differ from branchings in only a few extra arcs. 
At first glance, one might expect this to be
achievable by simply guessing the extra arcs and solving the remaining
problem for branchings. However, we do not know whether such a reduction
is possible in the strict sense. Indeed, we had to take a slight detour
and modify the two matroids in a way that guarantees a control for the
interactions caused by the presence of high-in-degree nodes. As a
result, we got an algorithm that runs in time polynomial in the input
size: namely, there can be more than ${n-1 \choose k+1}$ relevant input
values for each of the $n$ nodes; so, the runtime of our algorithm is
less than cubic in the size of the input, supposing the local scores are
given explicitly. While this answers one question in the affirmative, it
also raises several further questions, some of which we give in the next
paragraphs.

Our complexity analysis relied on a result concerning the general weighted matroid
intersection problem. Do significantly faster algorithms exist when restricted to our two specific matroids? One might expect such algorithms exist, since the related problem for
branchings can be solved in $O(n^2)$ time by the algorithm of
\cite{Tarjan77}.

Even if we could solve the matroid intersection problem faster, our algorithm would remain practical only for very small values of $k$. 
Can one find an optimal $k$-branching significantly faster, especially if allowing every node to have at most two parents? As the current algorithm makes
around $n^{3k}$ mutually overlapping guesses, there might be a way to
considerably reduce the time complexity. Specifically, we ask whether
the restricted problem is fixed-parameter tractable with respect to
the parameter $k$, that is, solvable in $O(f(k)p(n))$ time for some
computable function $f$ and polynomial $p$~\cite{DowneyFellows99}.
The fixed-parameter algorithm given in Section~\ref{sec:fpt} can be seen as a
first step towards an answer to this question. Can we find
other restrictions under which the
$k$-branching problem becomes fixed-parameter tractable?

Can we use a similar approach for the more general $k$-node branching problem,
i.e., is there a polynomial time algorithm for the $k$-node branching
problem for every fixed $k$? Likewise, we do not know whether the problem is easier or harder for
polytrees than for general DAGs: Do similar techniques apply to
finding maximum-score DAGs that can be turned into branchings by
deleting some $k$ arcs?


\subsubsection*{Acknowledgments}
Serge Gaspers, Sebastian Ordyniak, and Stefan Szeider acknowledge support from the European Research Council (COMPLEX REASON, 239962). 
Serge Gaspers acknowledges support from the Australian Research Council (DE120101761).
Mikko Koivisto acknowledges the support from the Academy of Finland (Grant 125637).
Mathieu Liedloff acknowledges the support from the French Agence Nationale de la Recherche (ANR AGAPE ANR-09-BLAN-0159-03).


\end{document}